\documentclass{amsart}

\usepackage{amssymb,stmaryrd,mathrsfs,dsfont}

\theoremstyle{plain}
\newtheorem{theorem}{Theorem}[section]
\newtheorem{lemma}[theorem]{Lemma}
\newtheorem{proposition}[theorem]{Proposition}
\newtheorem{corollary}[theorem]{Corollary}

\theoremstyle{definition}

\newtheorem{example}[theorem]{Example}

\theoremstyle{remark}
\newtheorem{remark}[theorem]{Remark}

\input xy
\xyoption{all}

\def\au{\mathcal{A}}

\begin{document}
\title[The Holonomy Decomposition of CSFA]{The Holonomy Decomposition of Circular Semi-Flower Automata}
\author[S. N. Singh]{Shubh Narayan Singh}
\address{Department of Mathematics, Central University of Bihar, Patna, India}
\email{shubh@cub.ac.in}
\author[K. V. Krishna]{K. V. Krishna}
\address{Department of Mathematics, IIT Guwahati, Guwahati, India}
\email{kvk@iitg.ac.in}


\begin{abstract}
Eilenberg's holonomy decomposition is useful to ascertain the structural properties of automata. Using this method, Egri-Nagy and Nehaniv characterized the absence of certain types of cycles in automata. In the direction of studying the structure of automata with cycles, this work focuses on a special class of semi-flower automata and establish the holonomy decompositions of certain circular semi-flower automata.
\end{abstract}



\maketitle

\section*{Introduction}

Usefulness of a decomposition method for any given system does not require any justification. The primary decomposition theorem due to Krohn and Rhodes has been considered as one of the fundamental results in the theory of automata and monoids \cite{krohn65}. Eilenberg has given a slight generalization of the primary decomposition called the holonomy decomposition \cite{eilenberg76}. Here, Eilenberg established that every finite transformation monoid divides a wreath product of its holonomy permutation-reset transformation monoids. The holonomy decomposition is also used to study the structural properties of certain algebraic structures \cite{holcombe80,kvk07}. The holonomy method appears to be relatively efficient and has been implemented computationally \cite{attila05,attila04}. One can use the computer algebra package, SgpDec \cite{nagy10} to obtain the holonomy decomposition of a given finite transformation monoid.

In order to ascertain the structure of an automaton, the holonomy decomposition considers the monoid of the automaton and looks for groups induced by the monoid permuting some set of subsets of the state set. These groups are called the holonomy groups, which are the building blocks for the components of the decomposition. Using holomony decomposition, Egri-Nagy and Nehaniv characterized the absence of certain types of cycles in automata. In fact, they proved that an automaton is algebraically cyclic-free if and only if the holonomy groups are trivial \cite{nagy05}. On the other hand, the structure of automata with cycles is much more complicated.

In the direction of studying the structure of automata with cycles, this work concentrates on a special class of semi-flower automata. A semi-flower automaton (SFA) is a trim automaton with a unique initial state that is equal to a unique final state in which all the cycles shall pass through the initial-final state \cite{giam07,shubh12}. Using SFA, the rank and intersection problem of certain submonoids of a free monoid have been studied \cite{giam08,singh11,singh12}.

In this paper, we consider circular SFA classified by their bpi(s) -- branch point(s) going in -- and obtain the holonomy decompositions for circular SFA with at most two bpis. The main work of the paper is presented in Section 2. Before that, in Section 1, we present some preliminary concepts and results that are used in this work. Finally, Section 3 concludes the paper.

\section{Preliminaries}

This section has two subsections on the holonomy decomposition and automata to present a necessary background material on these topics.

\subsection{The Holonomy Decomposition}

In this subsection, we provide a brief details on the holonomy decomposition which will be useful in this paper. For more details one may refer \cite{nehaniv05,attila05,eilenberg76}.

We first fix our notation regarding functions. We write the argument of a function $f: X \longrightarrow X$ on its left so that $xf$ is the value of the function $f$ at the argument $x$. The \emph{rank} of the function $f$, denoted by $\mbox{rank}(f)$, is the cardinality of its image set $Xf$.  Further, the composition of functions is designated by concatenation, with the leftmost function understood to apply first so that $xfg = (xf)g$.

A pair $(P, M)$ with a nonempty finite set $P$ and a monoid $M$ is called a \emph{transformation monoid} if there is an embedding $\phi: M \hookrightarrow \mathscr{T}(P)$, where $\mathscr{T}(P)$ is the monoid of all functions on $P$  with respect to the composition. A transformation monoid $(P, M)$ is called \emph{transformation group} if $M$ is a group. Let us denote the action of $m \in M$ on $p \in P$ as $pm$, rather than $p(m\phi)$. For $p \in P$, let $\widehat p$ be the constant function on $P$ which takes the value $p$, i.e. $q\widehat{p} = p$, $\forall q \in P$.
The \emph{closure} of a transformation monoid $(P, M)$, denoted by $\widehat{(P, M)}$, is defined as $(P, \widehat M)$, where $\widehat M$ is the monoid generated by $\displaystyle M \cup \bigcup_{p \in P}\{\widehat{p}\}$.

The \emph{skeleton space} $\mathscr{J}$ of a transformation monoid $(P, M)$ is \[\Big\{Pm\; \Big|\; m \in M\Big\}  \cup \displaystyle \bigcup_{p \in P}\Big\{\{p\}\Big\}\] with the preorder $\leq$ defined by: for $R, S \in \mathscr{J}$, $R \leq S$ if and only if $R\subseteq Sm$ for some $m \in M$. Consequently, we define an equivalence relation $\sim$ on $\mathscr{J}$ by \[R\sim S \;\;\mbox{if and only if}\;\; R \leq S\;\; \mbox{and}\;\; S \leq R.\]
We shall write $\mathscr{J}_i$ to denote the set of all elements of the skeleton space $\mathscr{J}$ of cardinality $i$ (for $i \geq 1$), that is,
\[\mathscr{J}_i = \Big\{T \in \mathscr{J}\;\Big|\; |T| = i \Big\}.\]
For $T \in \mathscr{J}$, put \[K(T) = \{m \in M\;|\; Tm = T\},\] the nonempty set of all elements of $M$ that act as permutations on the set $T$.
For $T \in \mathscr{J}$ with $|T|> 1$, the \emph{paving of $T$}, denoted by $B(T)$, is defined to be the set of maximal elements (with respect to set inclusion) of $\mathscr{J}$ that are contained in $T$, that is, \[B(T) = \{R \in \mathscr{J}\ |\ R \subsetneq T \mbox{ and if } S \in \mathscr{J} \mbox{ with } R \subseteq S \subseteq T \mbox{ then } S = R \mbox{ or } S = T\}.\]
Further, the set $G(T)$ of all the distinct permutations of $B(T)$ induced by the elements of $K(T)$ is called the \emph{holonomy group} of $T$ in $(P, M)$, and $(B(T), G(T))$ is a transformation group. It can be observed that, for $T_1, T_2 \in \mathscr{J}$ with $|T_1|, |T_2| > 1$, if $T_1 \sim T_2$, then $(B(T_1), G(T_1))$ is isomorphic to $(B(T_2), G(T_2))$.

The \emph{holonomy decomposition theorem} due to Eilenberg states that every finite transformation monoid divides a wreath product of its holonomy permutation-reset transformation monoids, as presented in the following:

\begin{theorem}[\cite{eilenberg76}]
If $(P, M)$ is a finite transformation monoid of height $n$, then
\[(P, M)  \prec \widehat{\mathscr{H}_1} \wr \widehat{\mathscr{H}_{2}} \wr \ldots \wr \widehat{\mathscr{H}_n},\]
where, for $1\leq i\leq n$, \[\mathscr{H}_i = \left(\prod_{j = 1}^{k_i} B(T_{ij}), \prod_{j = 1}^{k_i}G(T_{ij})\right),\] in which $k_i$ is the number of equivalence classes at height $i$ and $\{T_{ij}\ |\ 1\leq j \leq k_i\}$ is the set of representatives of equivalence classes at height $i$.
\end{theorem}

\subsection{Automata}

This subsection is devoted for essential preliminaries on automata and monoids.  For more details one may refer \cite{berstel85,giam07,shubh12}.

Let $A$ be a finite set called an \emph{alphabet} with its elements as \emph{letters}. The free monoid over $A$ is denoted by $A^*$ whose elements are called words, and $\varepsilon$ denotes the empty word -- the identity element of $A^*$.

An \emph{automaton} $\au$ over an alphabet $A$ is a quadruple $\au = (Q, I, T, \mathcal{F})$, where $Q$ is a finite set called the set of \emph{states}, $I$ and $T$ are subsets of $Q$ called the sets of \emph{initial} and \emph{final} states, respectively, and $\mathcal{F}\subseteq Q\times A\times Q$ called the set of \emph{transitions}. Clearly, by denoting the states as vertices/nodes and the transitions as labeled arcs, an automaton can be represented by a digraph in which initial and final states shall be distinguished appropriately. A \emph{path} in $\au$ is a finite sequence of consecutive arcs in its digraph. For $p_i \in Q$ ($0\le i \le k$) and $a_j \in A$ ($1 \le j \le k$), let
\[p_0 \xrightarrow{a_1} p_1 \xrightarrow{a_2} p_2 \xrightarrow{a_3} \cdots \xrightarrow{a_{k-1}} p_{k-1} \xrightarrow{a_k} p_k\] be a path, say $P$, in $\au$. The word $a_1\cdots a_k \in A^*$ is the \emph{label of the path} $P$. A \emph{null path} is a path from a state to itself labeled by $\varepsilon$. A path that starts and ends at the same state is called as a \emph{cycle}, if it is not a null path.

In an automaton $\au$ over $A$, a state $q$ is called a \emph{branch point going in}, in short \emph{bpi}, if the number of transitions coming into $q$ (i.e. the indegree of $q$ -- the number of arcs coming into $q$ -- in the digraph of $\au$) is at least two. We write $BPI(\au)$ to denote the set of all bpis of $\au$. A state $q$ of $\au$ is \emph{accessible} (respectively, \emph{coaccessible}) if there is a path from an initial state to $q$ (respectively, a path from $q$ to a final state). An automaton is said to be \emph{trim} if all the states of the automaton are accessible and coaccessible.

An automaton is called a \emph{semi-flower automaton} (in short, SFA) if it is a trim automaton with a unique initial state that is equal to a unique final state such that all the cycles visit the unique initial-final state. If an automaton $\au = (Q, I, T, \mathcal{F})$ is an SFA, we denote the initial-final state by $q_0$. In which case, we simply write $\au = (Q, q_0, q_0, \mathcal{F})$.

An automaton is \emph{deterministic} if it has a unique initial state and there is at most one transition defined for a state and an input letter. An automaton is \emph{complete} if there is at least one transition defined for a state and an input letter.

Let $\au = (Q, q_0, T, \mathcal{F})$ be a complete and deterministic automaton over $A$. Since there is a unique transition for each pair of a state and an input letter,  we define a function $\delta : Q \times A \longrightarrow Q$ by \[\delta(p, a) = q \;\;\mbox{if and only if}\;\; (p,a,q) \in \mathcal{F}.\] We can inductively extend the function for words by, for all $u \in A^*, a \in A$ and $q \in Q$, \[\delta(q,\varepsilon) = q,\; \mbox{ and }\; \delta(q, au) = \delta(\delta(q,a), u).\] We write $qu$ instead of $\delta(q, u)$. There is a natural way to associate a finite monoid to a complete and deterministic automaton $\au$. For each $x \in A^*$, we define a function $\overline{x} : Q \longrightarrow Q$ by $q \overline{x} = q x$, for all $q \in Q$.
The set of functions, $M(\au) = \{\overline{x} \;|\;x \in A^*\}$, forms a monoid under the composition of functions, called the \emph{monoid} of $\au$. Clearly, the monoid $M(\au)$ is generated by the functions defined by the letters of $A$. Further, for all $x, y \in A^*$, we have $\overline{xy} = \overline{x} \;\overline{y}$ and $\overline{\varepsilon}$ is the identity function on $Q$.

Let $X = \{p_1, \ldots, p_r\}$ be a finite set and $Y \subseteq X$. A \emph{$Y$\!\!-cycle}  is a permutation $f_Y$ on $X$ such that $f_Y$ induces a cyclic ordering on $Y$ (= $\{p_{i_1}, \ldots, p_{i_s}\}$, say) and $f_Y$ is identity on $X \setminus Y$, i.e., for $1 \le j < s$ and $p \in X \setminus Y$, \[p_{i_j}f_Y = p_{i_{j+1}}, \; p_{i_s}f_Y = p_{i_1},\; \mbox{ and }\; pf_Y = p.\] A \emph{circular permutation} on $X$ is an $X$\!-cycle. It is well known that for every permutation $f$ on $X$, there exists a partition $\{Y_i\}_{i \in \{1, 2, \ldots, t\}}$ of $X$ such that $f = f_{Y_1}f_{Y_2}\cdots f_{Y_t}$, a composition of (disjoint) $Y_i$-cycles.

A complete and deterministic automaton $\au$ over $A$ is said to be a \emph{circular automaton} if there exists $a \in A$ such that $\overline{a}$ is a circular permutation. Circular automata have been studied in various contexts. Pin proved the \v{C}ern\'{y} conjecture for circular automata with a prime number of states \cite{pin78}. Dubuc showed that the \v{C}ern\'{y} conjecture is true for circular automata \cite{dubuc98}.

In order to investigate the holonomy decomposition of circular semi-flower automata, in this paper we consider these automata classified by their number of bpis and complete the task for the automata with at most two bpis.

\section{Main Results}

We present our results of the paper in three subsections. In Subsection 2.1, we obtain some properties of circular semi-flower automata (CSFA) which are useful in the present work. Then, we investigate the holonomy decomposition of CSFA with at most one bpi and two bpis in subsections 2.2 and 2.3, respectively.

In what follows, $\au = (Q, q_0, q_0, \mathcal{F})$ is a complete and deterministic automaton over an alphabet $A$ such that $|Q| = n$. Further, for $m\geq 1$, $\mathscr{C}_m$ denotes a transformation group $(X, C_m)$, for some set $X$ with $|X| = m$ and  $C_m$ is the cyclic group generated by a circular permutation on $X$.

\subsection{Circular SFA}

In this subsection, first we ascertain that there is a unique circular permutation induced by the input symbols of CSFA and then we proceed to obtain certain properties pertaining to the bpis of CSFA.

\begin{proposition}\label{c3.l.ucp}
Let $\au$ be an SFA over $A$ and $a, b \in A$.
\begin{enumerate}
\item[(i)] If $\overline{a}$ is a permutation on $Q$, then $\overline{a}$ is a circular permutation on $Q$.
\item[(ii)] If $\overline{a}$ and $\overline{b}$ are permutations on $Q$, then $\overline{a} = \overline{b}$.
\end{enumerate}
\end{proposition}

\begin{proof}\
\begin{enumerate}
\item[(i)] Write $\overline{a} = f_{Q_1}\cdots f_{Q_t}$, a composition of $Q_i$-cycles for some partition $\{Q_i\}_{i \in \{1, \ldots, t\}}$ of $Q$. Let $q_0 \in Q_r$, for some $r$. If $Q_r = Q$, then $t = r = 1$ so that $\overline{a}$ is a circular permutation. Otherwise, there exist $q \in Q \setminus Q_r$ and $s \in \{1, \ldots, t\}$ such that $q \in Q_{s}$. Note that the $Q_{s}$-cycle induces a cycle in the digraph of $\au$ which does not pass through the state $q_0$; a contradiction.
\item[(ii)] On the contrary, let us assume that $\overline a \ne \overline b$. From part (i), the permutations $\overline a$ and $\overline b$ are circular permutations on $Q$. Let cyclic orderings on $Q$ with respect to $\overline a$ and $\overline b$ be as shown below.
    \begin{align*}
\overline a & : q_0, q_{i_1}, q_{i_2},\ldots, q_{i_{n-1}}\\
\overline b & : q_0, q_{j_1}, q_{j_2},\ldots, q_{j_{n-1}}
\end{align*}
Since $\overline a \neq \overline b$, let $k$ be the least number such that $q_{i_k} \neq q_{j_k}$. Note that there exists $s>k$ such that $q_{i_k} = q_{j_s}$ and also there exists $r>k$ such that $q_{j_k} = q_{i_r}$. Now, the path
\begin{equation*}
q_{i_k}\xrightarrow[]{a^{r-k}}q_{i_r} = q_{j_k}\xrightarrow[\hspace{7pt}\hspace{7pt}]{b^{s-k}} q_{j_s} = q_{i_k}
\end{equation*}
is a cycle labeled by $a^{r-k} b^{s-k}$. Clearly, this cycle does not pass through the initial-final state $q_0$; a contradiction.
\end{enumerate}
\end{proof}

\begin{corollary}
If $\au$ is a CSFA, then there is a unique circular permutation induced by the input symbols of $\au$.
\end{corollary}

\begin{proposition}\label{c3.l.nobpi}
Let $\au$ be an SFA over $A$; then, \[BPI(\au) = \varnothing \Longleftrightarrow |A| = 1.\]
\end{proposition}

\begin{proof}
In an $n$-state complete and deterministic automaton over $A$, \[\text{the total indegree of all states = the total number of transitions} = n|A|.\]
Since $\au$ is accessible, indegree of each state is at least one. Consequently,
\[BPI(\au) = \varnothing \Longleftrightarrow \text{the total indegree of all states }= n \Longleftrightarrow |A| = 1.\]
\end{proof}

Hereafter $\au$ is further assumed to be a CSFA. For the rest of the paper, we fix the following regarding $\au$. Assume $a \in A$ induces a circular permutation $\overline{a}$ on the state set $Q$ of $\au$. Accordingly, \[\overline a: q_0, q_1, \ldots, q_{n-1}\] is the cyclic ordering on $Q$ with respect to $\overline{a}$.

\begin{proposition}\label{c3.l.q0isbpi}
If $\au$ has at least one bpi, then its initial-final state is always a bpi.
\end{proposition}

\begin{proof}
Since $\au$ has at least one bpi, by Proposition \ref{c3.l.nobpi}, we have $|A| \ge 2$. We claim that $q_{n-1}\overline{b} = q_0$, for all $b \in A$, so that $q_0$ is a bpi. Let us assume the contrary, i.e. $q_{n-1}\overline{c} \ne q_0$, for some $c \in A$. Since $\au$ is complete and deterministic, $q_{n-1}\overline{c} = q_i$, for some $i$ (with $1 \le i < n$). Note that $q_i\overline{a^{n-i-1}c} = q_i$. Thus, we have a cycle in $\au$ from $q_i$ to $q_i$ labeled by $a^{n-i-1}c$ that does not visit $q_0$. This is a contradiction. Hence, $q_{n-1}\overline{b} = q_0$, for all $b \in A$.
\end{proof}

\begin{proposition} \label{c3.l.atmostbpi}
For $1 \leq m < n$, if $|BPI(\au)| = m$, then any non-permutation in $M(\au)$ has rank at most $m$.
\end{proposition}

\begin{proof}
In view of Proposition \ref{c3.l.nobpi}, we have $|A| > 1$. It is clear that $\overline a$ contributes one to the indegree of each state of $\au$. For $b \in A \setminus \{a\}$, if $|Q\overline{b}| > m$, then $|BPI(\au)| > m$; a contradiction.  Thus, $|Q\overline{b}| \le m$ for all $b \in A \setminus \{a\}$.   Now, for $x \in A^*$, if $\overline{x}$ is a non-permutation, then $x$ contains a symbol $b \in A \setminus \{a\}$. Hence, the rank of $\overline{x}$ is at most $m$.
\end{proof}

In view of Proposition \ref{c3.l.q0isbpi}, we have the following corollary of Proposition \ref{c3.l.atmostbpi}.

\begin{corollary}\label{c3.c.Qbq0}
If $\au$ has a unique bpi, then $Q\overline b = \{q_0\}$, for all $b \in A\setminus \{a\}$.
\end{corollary}

\subsection{CSFA with at most one bpi}

In this subsection, we obtain the holonomy decomposition of CSFA with at most one bpi. We first observe that the holonomy decomposition of SFA with no bpis follows from the general case of permutation SFA. An automaton is a \emph{permutation automaton} if the function induced by each input symbol is a permutation on the state set \cite{thi68}. Clearly, an automaton is a permutation automaton if and only if its monoid is a group.

By Proposition \ref{c3.l.ucp}, we have the following proposition which also provides the holonomy decomposition of a permutation SFA.

\begin{proposition}\label{c3.p.hdpa}
If $\au$ is a permutation SFA, then $M(\au)$ is a cyclic group. \break Further, \[(Q, M(\au)) \prec \widehat{\mathscr{C}_n}.\]
\end{proposition}

Now, we investigate the holonomy decomposition of CSFA with no bpis. If $\au$ is an SFA with no bpis,
then $|A| = 1$, say $A = \{a\}$ (cf. Proposition \ref{c3.l.nobpi}). Note that the function $\overline a$ is a circular permutation on $Q$. Thus, $\au$ is a circular as well as permutation SFA. Hence, by Proposition \ref{c3.p.hdpa}, we have the following theorem.

\begin{theorem}
Let $\au$ be an SFA with no bpis, then \[(Q, M(\au)) \prec \widehat{\mathscr{C}_n}.\]
\end{theorem}

Now, we present the holonomy decomposition of CSFA with a unique bpi in the following theorem.

\begin{theorem}
If $\au$ is a CSFA with a unique bpi, then \[(Q, M(\au)) \prec \widehat{\mathscr{C}_n}.\]
\end{theorem}

\begin{proof}
By Corollary \ref{c3.c.Qbq0}, we have $Q\overline b = \{q_0\}$, for all $b\in A\setminus \{a\}$. This implies that $\overline b = \overline c$, for all $b, c\in A\setminus \{a\}$. Thus, $M(\au)$ is generated by the set $\{\overline a , \overline b\}$.

For $\overline x \in M(\au)$, by Proposition \ref{c3.l.atmostbpi}, we have either $|Q \overline x| = n$ or $|Q \overline x| = 1$. Consequently, the skeleton space of $(Q, M(\au))$ is \[\mathscr{J} = \{Q\} \cup \mathscr{J}_1.\] Note that
\[K(Q) =\Big\{\overline{a^i}\;\Big|\; 1\leq i \leq n\Big\}\hspace{0.5cm}\mbox{and}\hspace{0.5cm} B(Q) = \mathscr{J}_1.\]

Clearly, $|B(Q)| = n$ and the holonomy group of $Q$ is \[G(Q) =\Big\{\check{\overline{a^i}}\;\Big|\; 1\leq i \leq n\Big\},\]
where each element $\check{\overline {a^i}}$ is the permutation on $B(Q)$ induced by the corresponding element $\overline{a^i} \in K(Q)$.  For $1 \leq i \leq n$, since $\overline {a^i} = {\overline a}^i$, we have  $\check{\overline {a^n}}= \check{(\overline a^n)} = \check{\overline\varepsilon}$. This implies that the holonomy group $G(Q)$ is a cyclic group of order $n$ generated by $\check{\overline a}$. Consequently, we have \[(Q, M(\au)) \prec \widehat{\mathscr{C}_n}.\]
\end{proof}

\subsection{CSFA with two bpis}

In this subsection, we investigate the holonomy decomposition of CSFA with two bpis. Here, $\au$ denotes a CSFA with two bpis. By Proposition \ref{c3.l.q0isbpi}, the initial-final state $q_0$ of $\au$ is a bpi. Let $q_m$, where $1 \le m < n$, be the other bpi of $\au$ so that $BPI(\au) = \{q_0, q_m\}$. Note that, by Proposition \ref{c3.l.nobpi}, we have $|A| \ge 2$.

\begin{lemma}\ \label{c3.l.r2}
\begin{enumerate}
\item[\rm(i)] For $b \in A$, if $\mbox{rank}(\overline b) = 2$, then $Q\overline b = BPI(\au)$.
\item[\rm(ii)] There exists a symbol $b \in A$ such that $Q\overline b = BPI(\au)$.
\end{enumerate}
\end{lemma}

\begin{proof} We first note that $\overline a$ contributes one to the indegree of each state in $Q$.
Since $BPI(\au) = \{q_0, q_m\}$, we have $Q \overline b \subseteq \{q_0, q_m\}$,  for all $b \in A\setminus \{a\}$.
\begin{enumerate}
\item[\rm(i)] Straightforward from the above statement.
\item[\rm(ii)]Let us assume that $Q\overline b \neq \{q_0, q_m\}$, for all $b \in A\setminus \{a\}$. Then, for $b \in A\setminus \{a\}$, either $Q\overline b =\{q_0\}$ or $Q\overline b = \{q_m\}$. For some $b \in A\setminus \{a\}$, if $Q\overline b =\{q_m\}$, then there is a loop at $q_m$; which is not possible. Consequently, for all $b \in A\setminus \{a\}$, $Q\overline b =\{q_0\}$. This implies $BPI(\au) = \{q_0\}$; a contradiction.  Hence, there exists $b \in A$ such that $Q\overline b = BPI(\au)$.
\end{enumerate}
\end{proof}

The following lemma provides the skeleton space of the transformation monoid $(Q, M(\au))$.

\begin{lemma}\label{c3.l.sshentes}
The skeleton space of the transformation monoid $(Q, M(\au))$ is given by \[\mathscr{J} = \{Q\} \cup \mathscr{J}_2 \cup \mathscr{J}_1,\] where \[\mathscr{J}_2 = \Big\{\{q_0, q_m\}\overline{a^i}\;\Big|\;1\leq i \leq n \Big\}.\]
\end{lemma}

\begin{proof}
In view of Proposition \ref{c3.l.atmostbpi}, other than $Q$ and singletons, the skeleton space $\mathscr{J}$ can have some sets of size two. Thus, it is sufficient to determine $\mathscr{J}_2$.

By Lemma \ref{c3.l.r2}(ii), there exists an input symbol $b \in A\setminus \{a\}$ such that
$Q\overline b = \{q_0, q_m\}$. Therefore, for all $1\leq i \leq n$, the image set \[Q\overline{ba^i} = \{q_0, q_m\}\overline{a^i} \in \mathscr{J}_2.\] Thus, we have \[\Big\{\{q_0, q_m\}\overline{a^i}\;\Big|\;1\leq i \leq n \Big\}\subseteq \mathscr{J}_2.\]

Let us assume that $Q\overline w \in \mathscr{J}_2$, for some $w \in A^*$. Then $w$ is of the form
\[w = a^{i_1} b_1 a^{i_2} b_2  \cdots a^{i_{k}} b_k a^{i_{k+1}},\]
for $i_j \geq 0$ ($1 \le j \le k+1$) and $b_i \in A$ ($1\leq i \leq k$) such that the rank of each function $\overline{b_i}$ is two (cf. Proposition \ref{c3.l.atmostbpi}). Write $w =  a^{i_1} b_1ub_k a^{i_{k+1}}$, where $u = a^{i_2} b_2  \cdots a^{i_{k}}$. Since rank($\overline{b_1ub_k}$) = rank($\overline{b_k}$) = 2,  we have
\[Q\overline{b_1ub_k} = Q\overline b_k = \{q_0, q_m\},\] by Lemma \ref{c3.l.r2}(i). Consequently, \[Q\overline{w} = Q\overline{a^{i_1} b_1ub_ka^{i_{k+1}}} = \{q_0, q_m\}\overline{a^{i_{k+1}}}.\]
Hence, \[\mathscr{J}_2 = \Big\{\{q_0, q_m\}\overline{a^i}\;\Big|\;1\leq i \leq n \Big\}.\]
\end{proof}

\begin{remark}\label{c3.r.j2}
As shown in Example \ref{c3.e.even}, the cardinality of $\mathscr{J}_2$ is not necessarily $n$.
\end{remark}

\begin{example}\label{c3.e.even}
The automaton $\au$ given in Figure \ref{fig3.2} is a CSFA with $BPI(\au) = \{q_0, q_2\}$, and $|Q| = 4$ .
Here, $Q\overline b = \{q_0, q_2\}$ and  we observe that \[\{q_0, q_2\}\overline a = \{q_1, q_3\}, \{q_0, q_2\}\overline{a^2} = \{q_0, q_2\},\; \mbox{ so that }\; |\mathscr{J}_2| = 2.\]
\begin{figure}[t]
\entrymodifiers={++[o][F-]} \SelectTips{cm}{}
\[\xymatrix{*\txt{} & *++[o][F=]{q_0} \ar[dr]^a \ar@(r,u)[]_b& *\txt{}\\
q_3 \ar[ur]^a \ar@/^1.5pc/[ur]^b & *\txt{} & q_1 \ar[dl]^a \ar@/^1.5pc/[dl]^b\\
 *\txt{} & q_2 \ar[ul]^a \ar[uu]^b & *\txt{} }\]
\caption{A CSFA $\au$ with two bpis}
\label{fig3.2}
\end{figure}
\end{example}

\begin{lemma}\label{c3.l.q0xqk}
There exists $x \in A^*$ such that $q_0\overline{x} = q_m$ and $q_m\overline{x} = q_0$.
\end{lemma}

\begin{proof}
If there exists $b \in A \setminus \{a\}$ such that $q_0\overline{b} \ne q_0$, then clearly $x = b$ will serve the purpose. Otherwise, we have
$q_0\overline{b} = q_0$, for all $b \in A \setminus \{a\}$. However, by Lemma \ref{c3.l.r2}, there exist a symbol $c \in A$ such that $Q\overline{c} = \{q_0, q_m\}$. If $c = a$, then $Q = \{q_0, q_m\}$ and the result is straightforward.

Let us assume that $c \ne a$. Clearly, $q_0\overline{c} = q_0$ and there exists a state $q_i$ (with $1 \le i < m$) such that $q_i\overline{c} = q_m$. Let $t$ (with $1 \le t < m$) be the least number such that $q_t\overline{c} = q_m$. Choose $x = a^tc$ and observe that $q_0\overline{x} = q_m$.
We claim that $q_m\overline{x} = q_0$.

On the contrary, assume $q_m\overline{x} \ne q_0$. Then, $q_m\overline{x} = q_m$ so that there is a cycle from $q_m$ to $q_m$ labeled $x$. Thus, the cycle should pass through $q_0$.  Since $q_0\overline c = q_0$, there exist $t_1$ and $t_2$ ($1 \leq t_1, t_2 < t$) with $t_1 + t_2 = t$ such that
\[q_m \overline {a^{t_1}} = q_0 \;\;\mbox{and}\;\; q_0 \overline{a^{t_2}c} = q_m.\]
Note that $q_0 \overline{a^{t_2}c} = q_{t_2}\overline{c} = q_m$. This contradicts the choice of $t$, as $t_2 < t$. Thus, $q_m\overline
{x} = q_0$.
\end{proof}

\begin{theorem}\label{c3.t.ck2}
If $\au$ is a CSFA with $BPI(\au) = \{q_0, q_m\}$, then \[(Q, M(\au)) \prec  \widehat{\mathscr{C}_2} \wr \widehat{\mathscr{C}_r},\] where $r$ (with $1 \le r \le n$) is the smallest number such that $\{q_0, q_m\}\overline{a^r} = \{q_0, q_m\}$.
\end{theorem}

\begin{proof}
From Lemma \ref{c3.l.sshentes}, the skeleton space of $(Q, M(\au))$ is
\[\mathscr{J} = \{Q\} \cup \mathscr{J}_2 \cup \mathscr{J}_1\] in which all the elements of $\mathscr{J}_2$ are equivalent to each other.

For $1 \le i \le n$, note that $\overline{a^i}$ permutes the elements of $Q$ and, for $x \in A^*$, if $\overline{x} \ne \overline{a^i}$, then $\overline{x}$ is not a permutation on $Q$ (cf. Proposition \ref{c3.l.ucp}). Consequently,
\[K(Q) = \Big\{\overline{a^i}\;\Big|\; 1\leq i \leq n \Big\}.\] Since all the elements of $\mathscr{J}_2$ are maximal in $Q$, we have $B(Q) = \mathscr{J}_2$. Let $r$ (with $1 \leq r \leq n$) be the smallest integer such that $\{q_0, q_m\}\overline{a^r} = \{q_0, q_m\}$ so that $|B(Q)| = r$. Consequently, the holonomy group \[G(Q) =\Big\{\check{\overline{a^i}}\;\Big|\; 1\leq i \leq r \Big\},\] where each function $\check{\overline{a^i}}$ is a permutation on $B(Q)$ induced by the corresponding function $\overline{a^i} \in K(Q)$. Since $\check{\overline{a^i}} = \check{\overline{a}}^i$, the holonomy group $G(Q)$ is  a cyclic group of order $r$ generated by $\check{\overline{a}}$. Thus, \[(B(Q), G(Q)) = \mathscr{C}_r.\]

Let $P = \{q_0, q_m\}$ be a representative in $\mathscr{J}_2$. Clearly, \[B(P) = \{\{q_0\}, \{q_m\}\}.\] By Lemma \ref{c3.l.q0xqk}, there exist $x \in A^*$ such that $q_0\overline{x} = q_m$ and $q_m\overline{x} = q_0$, so that $K(P) = \{\overline x, \overline{\varepsilon}\}$. Consequently, the holonomy group $G(P) = C_2$ and hence, \[(B(P), G(P)) = \mathscr{C}_2.\]

Thus, the holonomy decomposition of $\au$ is given by
\[(Q, M(\au))\prec \widehat{\mathscr{C}_2} \wr \widehat{\mathscr{C}_r}.\]
\end{proof}

\begin{corollary}\label{c3.c.oddhd}
Let $n$ be an odd number; if $\au$ is a CSFA with two bpis such that $|Q| = n$, then \[(Q, M(\au)) \prec  \widehat{\mathscr{C}_2} \wr \widehat{\mathscr{C}_n}.\]
\end{corollary}

\begin{proof}
From Theorem \ref{c3.t.ck2}, we have \[(Q, M(\au))\prec \widehat{\mathscr{C}_2} \wr \widehat{\mathscr{C}_r},\] where $r$ (with $1 \le r \leq n$) is the smallest number such that $\{q_0, q_m\}\overline{a^r} = \{q_0, q_m\}$. We claim that  $r = n$. If $r < n$, since $\{q_0, q_m\}\overline{a^r} = \{q_0, q_m\}$ and $\overline a$ is a circular permutation on $Q$, it follows that $q_0\overline{a^r} = q_m$, and $q_m\overline{a^r} = q_0$. This implies that $q_0\overline{a^{2r}} = q_0$ with $1 < 2r < 2n$. Therefore, $2r = n$; a contradiction. Hence, \[(Q, M(\au))\prec \widehat{\mathscr{C}_2} \wr \widehat{\mathscr{C}_n}.\]
\end{proof}

\section{Conclusion}

In this work, we have initiated the investigations on the holonomy decomposition of circular semi-flower automata (CSFA), classified by their number of bpis. In fact, we have ascertained the holonomy decompositions of CSFA with at most two bpis. Our experiments for the holonomy decomposition of CSFA with more than two bpis over a numerous examples exhibit that their structure is much more complicated. However, we feel that the approach adopted in this paper may be useful to target the holonomy decomposition of CSFA with arbitrary number of bpis. In general, one can look for the holonomy decomposition of SFA. There is a lot more to investigate on the structure of automata with cycles, as a more general problem.


\end{document}